\documentclass{amsart}
\usepackage{amssymb,latexsym,amsmath}
\usepackage{cite}

\newcommand{\ub}[1]{\underbrace{#1}}

\newcommand{\ourU}{\mathbb U}%
\newcommand{\mcU}{{\mycal U}}

\newcommand{\bel}[1]{\begin{equation}\label{#1}}

\newcommand{\ee}{\end{equation}}

\newcommand{\hyp}{\mycal S}

\newcommand{\N}{\mathbb N}
\newcommand{\R}{\mathbb R}
\newcommand{\ext}{\mathrm{ext}}
\newcommand{\mcM}{{\mycal M}}

\newcommand{\eq}[1]{(\ref{#1})}

\newcommand{\beaa}{\begin{eqnarray*}}
\newcommand{\eeaa}{\end{eqnarray*}}

\newtheorem{theorem}{Theorem}
\newtheorem{lemma}{Lemma}[section]
\newtheorem{proposition}[lemma]{Proposition}
\newtheorem{corollary}[lemma]{Corollary}

\DeclareFontFamily{OT1}{rsfs}{}
\DeclareFontShape{OT1}{rsfs}{m}{n}{ <-7> rsfs5 <7-10> rsfs7 <10-> rsfs10}{}
\DeclareMathAlphabet{\mycal}{OT1}{rsfs}{m}{n}

\begin{document}
\title{Construction of $N$-body initial data sets in general relativity}
\author{Piotr T. Chru\'{s}ciel} 
\address{University of Vienna, Vienna, Austria}
\email{piotr.chrusciel@univie.ac.at}
\author{Justin Corvino} 
\address{Department of Mathematics, Lafayette College, Easton, PA 18042 USA}
\email{corvinoj@lafayette.edu}
\author{James Isenberg}
\address{Department of Mathematics, University of Oregon, Eugene, 
OR 97403-5203, USA}
\email{isenberg@uoregon.edu}

\begin{abstract}
Given a collection of $N$ solutions of the $(3+1)$ vacuum Einstein constraint equations which are asymptotically Euclidean, we show how to construct a new solution of the constraints which is itself asymptotically Euclidean, and which contains specified sub-regions of each of
the $N$ given solutions.  This generalizes earlier work which handled the time-symmetric case, thus providing a construction of large classes of initial data for the many body problem in general relativity.
\end{abstract}

\maketitle

\section{Introduction.}

An important problem in any theory of gravitation is the description of the motion of many-body systems. Various approximation schemes have been proposed to analyze this in general relativity,  but no rigorous treatment has been provided thus far. A first step towards a solution of this question is to provide wide classes of initial data which solve  the general relativistic constraint equations, and which are relevant to the problem at hand.  There exists a rich family of initial data sets modeling isolated gravitational systems, but the non-linear nature of the constraints must be addressed when attempting to incorporate several such systems into a single one.  In recent work~\cite{ChCorIs} we have shown how this can be done by a gluing construction, under the restrictive hypothesis of time-symmetry. The aim of this work is to remove this restriction.

Given a Riemannian metric $g$ and a symmetric $(0,2)$-tensor
$K$ on an oriented three-manifold $M$, the Einstein constraints
map can be written in the form\footnote{Here and throughout
this paper, we use the Einstein summation convention.}
$$
 \Phi(g,K)= \left(
 \begin{array}{c}
 -2(g^{jk} K_{ij;k}-(K^j{}_j)_{;i})\\
 R(g)-K_{ij}K^{ij}+(K^i{}_i)^2
 \end{array}\right).
$$
The condition then for $(g,K)$ to be the first and second
fundamental forms of $M$ embedded in a Ricci-flat space-time is
that the vacuum constraint equations $\Phi(g,K)=(0,0)$ be
satisfied.

We recall that $(g,K)$ on the exterior $E$ of a ball in
$\mathbb{R}^3$ constitutes  an \emph{asymptotically Euclidean end} (to order
$\ell$) provided there are coordinates in which, for
multi-indices $|\alpha|\leq \ell+1$, $|\beta|\leq \ell$,
\begin{equation} \label{af}
|\partial^{\alpha}(g_{ij}-\delta_{ij})(\mathbf  x)|=O(|\mathbf  x|^{-|\alpha|-1})
 , \qquad
 |\partial^{\beta}K_{ij}(\mathbf  x)|=O(|\mathbf  x|^{-|\beta|-2}),
\end{equation}
where $\partial$ denotes the partial derivative operator.  Note
that throughout the rest of this work, we require that $\ell
\geq 2$. We say that $(M,g,K)$ is \emph{asymptotically
Euclidean} if $M$ is the union of a compact set and a finite
number of ends, all of  which are asymptotically Euclidean for
$(g,K)$ in the sense defined above.  One readily verifies that
every  asymptotically Euclidean (AE) end possesses a
well-defined energy-momentum vector $(m,\mathbf  p)$.

The starting point for constructing initial data for an $N$
body system is the choice of the bodies. Each body is
separately designated by the choice of an AE initial data set,
and by the specification of a fixed interior region in each
such AE solution.  To perform this construction, we also
specify a \emph{vacuum} AE end disjoint from the interior
region in each of the $N$ data sets.  We choose a collection of
points which roughly locate the $N$ bodies on a fiducial flat
background.  We then construct a new initial data set on a
manifold obtained by excising a neighborhood of infinity in
each chosen end, and then gluing these into a manifold
$M_{\rm{ext}}$ which is $\mathbb{R}^3$ with $N$ disjoint balls,
centered around the chosen points, removed.  In this way, the
bodies are made to interact, by gluing the various ends into a
fixed end.  We now construct initial data on the resulting
manifold which: i) is identical to the initial data from the
$N$ bodies away from $M_{\text{ext}}$, hence containing $N$
chosen regions isometric to the specified interior regions of
the bodies; ii) solves the vacuum constraints on
$M_{\text{ext}}$; iii) is identical to a space-like slice of a
Kerr space-time sufficiently far from the bodies; and iv) has
the centers of the bodies in a configuration which is a scaled
version of the chosen configuration, where the scale factor can
be chosen arbitrarily above a certain threshold.  We emphasize
that we preserve the original solutions away from a
neighborhood of the gluing region.  If the bodies are vacuum to
begin with, we produce a solution to the vacuum constraints
everywhere, but we can allow the initial data sets to be
non-vacuum away from the chosen AE end.

The construction actually produces a family of solutions depending on a parameter $\epsilon$; roughly speaking, the bodies cannot be arbitrarily close together, but must be separated by a distance above a certain threshold.  For each $\epsilon$, the distance between the bodies is on the order $O(\epsilon^{-1})$; the smaller $\epsilon$ is taken, the further apart the bodies will be, and thus the weaker will be their initial interaction.  We note that the energy-momentum four-vector of $M_{\rm{ext}}$ in the resulting solution tends to the sum of the four-vectors of the bodies as $\epsilon\rightarrow 0^+$.

 We now state our main theorem.  For clarity of exposition, we assume that the initial data sets are $C^{\infty}$-smooth, and refer the interested reader to \cite{ChDelay, CorvinoSchoen2} to formulate the case of a finite degree of regularity. 
 
 \begin{theorem}
 \label{main1} For each $k=1, \ldots, N$, let $(E_k,
g^k,K^k)$ be a three-dimensional AE end which solves the vacuum constraints $\Phi(g^k, K^k)=0$, with time-like energy-momentum four-vector $(m_k, \mathbf  p_k)$.
Let $U_k\subset E_k$ be a pre-compact neighborhood of the boundary $\partial E_k$, and let  $M_{\rm{ext}}= \mathbb{R}^3\setminus \bigcup\limits_{k=1}^N B_k$, where $\overline{B}_1, \ldots, \overline{B}_N$ are pairwise disjoint closed balls in $\mathbb{R}^3$.  There is an $\epsilon_0>0$ so that for $0<\epsilon<\epsilon_0$, there is a solution
$(M_{\rm{ext}},g_{\epsilon},K_\epsilon)$ of the vacuum constraint equations, with one AE end, containing the disjoint union $\bigcup\limits_{k=1}^N (U_k, g^k,K^k)$, so that the distances between distinct $U_k$ are $O(\epsilon^{-1})$.
Near infinity $(M_{\rm{ext}},g_{\epsilon},K_\epsilon)$ is isometric to a space-like slice of a Kerr metric, with the ADM energy-momentum $(m(g_{\epsilon}), \mathbf  p_{\epsilon})$ satisfying $\Big| m(g_{\epsilon})- \sum\limits_{k=1}^N m_k\Big|  < \epsilon$ and $\Big| \mathbf  p_{\epsilon}-\sum\limits_{k=1}^N \mathbf  p_k\Big | <\epsilon.$
\end{theorem}

The construction described here allows us to glue together any finite number of
asymptotically Euclidean ends which solve the vacuum
constraint equations,
and the construction is local near infinity in each end;
i.e., any given compact subset of the end can be realized
isometrically in the final metric $g_{\epsilon}$.  The local nature of the construction implies that we can allow the $N$ original
solutions to have multiple ends, and we can also allow nonzero matter fields supported outside a neighborhood of infinity in the chosen ends, as indicated below.

\begin{corollary}
Let $(M_k, g^k, K^k)$, $k=1, \ldots, N$, be three-dimensional initial data sets with vacuum AE ends $E_k \subset M_k$ of respective time-like ADM energy-momentum $(m_k, \mathbf  p_k)$.  Let $U_k\supset
M_k\setminus E_k$ be chosen subdomains with $E_k\cap U_k$
precompact.  There is an $\epsilon_0>0$ so that for $0<\epsilon<\epsilon_0$, there is
an initial data set $(M,g_{\epsilon}, K_{\epsilon})$ which contains a region $U$ isometric to $\bigcup\limits_{k=1}^N (U_k, g^k)$, for which $(M\setminus U, g_{\epsilon},K_{\epsilon})$ has one AE end, with the same properties as those of $(M_{\rm{ext}},g_{\epsilon},K_\epsilon)$ as in Theorem \ref{main1}.
 \label{maincor}
\end{corollary}

\section{Preliminaries}
 \label{sGd}

\subsection{Kerr-Schild Coordinates}
In order to establish certain estimates needed for the gluing
carried out below, it will be convenient to use the explicit
Kerr-Schild form of the Kerr metric (cf., e.g.,
\cite{Chandrasekhar84}):
$$
 g_{\mu\nu} = \eta_{\mu\nu} +\frac{2m{\tilde r}^3}{{\tilde r}^4+a^2 z^2} \theta_\mu \theta_\nu =\eta_{\mu \nu} + O\Big(\frac{m}{|\mathbf  x|}\Big)
 \;,
$$
where $\eta$ is the Minkowski metric (we are using the
signature $(-,+,+,+)$), with
\bel{25V.1}
 \theta_\mu dx^\mu = dx^0 - \frac 1 {{\tilde r}^2+a^2}\left[ {\tilde r} (xdx+ydy) + a (xdy-ydx)\right]-\frac z {\tilde r} dz
 \;,
\ee
and where ${\tilde r}$ is defined implicitly as the solution of
the equation
$$
 {\tilde r}^4 - {\tilde r}^2 ( x^2+y^2+z^2-a^2)-a^2 z^2 =0
 \;.
$$

We will continue to use the term \emph{Kerr-Schild coordinates}
for a coordinate system which has been obtained by a Lorentz
transformation from the above.

\subsection{Global Charges}
In establishing the main results, we use the fact that the Kerr space-times yield a family of initial data sets which admit coordinates (such as Kerr-Schild) with sufficient approximate parity symmetry that allows the definition of angular momentum $\mathbf  J$ and centre of mass $\mathbf  c$, in addition to the four-momentum $(m ,\mathbf  p)$.  Indeed the Kerr initial data sets we use satisfy the Regge-Teitelboim asymptotic conditions, which say that in suitable AE coordinates the following estimates also hold:
\begin{equation} \label{rt}
\Big|\partial^{\alpha}\Big(g_{ij}(\mathbf  x)-g_{ij}(-\mathbf  x)\Big)\Big|=O(|\mathbf  x|^{-|\alpha|-2})
 , \quad
\Big |\partial^{\beta}\Big(K_{ij}(\mathbf  x)+K_{ij}(-\mathbf  x)\Big)\Big|=O(|\mathbf  x|^{-|\beta|-3}).
\end{equation}
The space of initial data satisfying (\ref{rt}) is known to be dense in the space of vacuum AE data \cite{CorvinoSchoen2}.

Using such a coordinate system, we can compute the energy and linear and angular momenta using flux integrals at infinity ($d\sigma_e$ is Euclidean surface measure, $\nu$ is the Euclidean outward normal, and $r=|\mathbf  x|$):
\begin{eqnarray*}
m&=&\frac{1}{16\pi}\lim\limits_{R\rightarrow\infty}\int\limits_{\{r=R\}}
\sum\limits_{i,j}\left( g_{ij,i}-g_{ii,j}\right) \nu^j d\sigma_e \nonumber \\
p_i&=&\frac{1}{8\pi}\lim\limits_{R\rightarrow\infty}\int\limits_{\{r=R\}}
\sum\limits_{j}(K_{ij}-K^{\ell}{}_{\ell}g_{ij})\nu^j d\sigma_e \label{adm} \\
J_i&=&\frac{1}{8\pi}\lim\limits_{R\rightarrow\infty}
\int\limits_{\{r=R\}}\sum\limits_{j,k}
(K_{jk}-K^{\ell}{}_{\ell}g_{jk})Y^j_i\nu^k d\sigma_e \nonumber\\
 mc^{\ell} &=&\frac{1}{16\pi}\lim\limits_{R\rightarrow\infty}\int\limits_{\{r=R\}}
\Big[\sum\limits_{i,j} x^\ell \left( g_{ij,i}-g_{ii,j}\right) \nu^j  -
\sum\limits_{i}\big(g_{ik}\delta^{k\ell} \nu^i- g_{ii}\nu^\ell \big)\Big]  d\sigma_e . \nonumber
 \end{eqnarray*}
Note that in the last term, we can replace $g$ in the center
integrand by $(g-g_{\text{Eucl}})$.    Taken together, these
give a set of ten \emph{Poincar\'{e} charges} associated to the
end.  We emphasize that we do not impose condition (\ref{rt}) on the initial data for the bodies; rather, we show in Proposition \ref{kerr} that we can modify the given vacuum end of each body to a vacuum end in Kerr, preserving the data set away from the end.

We recall the relation between the charge integrals and the
constraints.  Indeed, these charges arise from integrating the
constraints against elements of the cokernel of the linearized
constraint operator.  By linearizing at the Minkowski data, we
have $\Phi(g_{\text{Eucl}}+h,K)= D\Phi(h, K)+ Q(h,K)$, where
$Q(h,K)=h\ast \partial^2 h + \partial h \ast \partial h +
\partial h \ast K$, where ``$\ast$" denotes some metric
contraction of the tensor product.  In Euclidean coordinates at
the Minkowski data, $D\Phi(h,K)=\Big( \sum\limits_j
(-2(K_{ij,j}-K_{jj,i})), \sum\limits_{i,j} (h_{ij, ij}-h_{ii,
jj})\Big).$  Thus for any vector and scalar pair $(Y,N)$ which
satisfies $D\Phi^*(Y,N)=(0,0)$, we have as a consequence of
integration by parts that
\begin{equation} \label{bdyint}
\int\limits_{\{ R_0\leq r\leq R\} } (Y,N)\cdot \Phi(g_{\text{Eucl}}+h, K) d\mu_e = \mathcal B(R)-\mathcal B(R_0) +  \int_{\{R_0\leq r \leq R\} } (Y, N)\cdot Q(h,K) d\mu_e
\end{equation}
where $d\mu_e$ is Euclidean volume measure, and where
$$\mathcal{B}(R)= \int_{r=R}
\Big((-2)Y^i(K_{ij}-K^{\ell}{}_{\ell}\delta_{ij})    +
N(\sum\limits_i (h_{ij,i}-h_{ii,j}) - \sum\limits_i
(N_{,i}h_{ij} - N_{,j} h_{ii}  \Big)\nu^j\; d\sigma_e .$$ By
letting $Y$ be a Euclidean Killing vector field, or letting $N$
be a constant or a coordinate function $x^\ell$, we can easily
relate $\mathcal B(R)$ to one of the above surface integrals
defining the ADM energy-momenta.

\subsection{Hamiltonian Formulation of the Poincar\'{e} Charges} \label{ham}
It will be convenient to use the Hamiltonian formulation of the
Poincar\'{e} charges, as we review now (see Appendix E of
\cite{ChDelay} and references therein).  Let $\hyp$ be a
three-dimensional spacelike hypersurface in a four-dimensional
Lorentzian space-time $(\mcM,\bar g)$. Suppose that $\mcM$
contains an open set $\mcU$ with a time coordinate $t$ (with
range not necessarily equal to $\R$), as well as a  ``radial''
coordinate $r\in[R,\infty)$, leading to local coordinate
systems $(t,r,v^A)$, with $(v^A)$ providing local coordinates
on a two-dimensional sphere. We further require that $\hyp\cap
\mcU=\{t=0\}$. Assume that the metric $\bar g_{\mu\nu}$
approaches the Minkowski metric $\eta_{\mu\nu}$ as $r$ tends to
infinity. Set $\Omega=dx^0\wedge dx^1\wedge dx^2 \wedge dx^3$,
and $dS_{\alpha\beta}= \Omega(\frac{\partial}{\partial
x^{\alpha}}, \frac{\partial}{\partial x^{\beta}}, \cdot,
\cdot)$. The Hamiltonian analysis of vacuum general relativity
in~\cite{ChAIHP} leads to the formula
$$ H(\hyp,\bar g,X)= \frac 12 \int_{\partial_{\infty}\hyp}
 \ourU^{\alpha\beta}dS_{\alpha\beta}
$$
for the Hamiltonian $ H(\hyp,\bar g,X)$ associated to the flow
of a vector field $X$, assumed to be a Killing vector field for
the Minkowski metric, where
\begin{eqnarray*}
  8\pi \ourU^{\nu\lambda}&= &
 \displaystyle{\frac{1}{ \sqrt{|\det \bar g|}}}
\bar g_{\beta\gamma}\mathring \nabla_\kappa
(|\det \bar g |\; \bar g^{\gamma[\nu}\bar g^{\lambda]\kappa})X^\beta +  \sqrt{|\det \bar g|}~\bar g^{\alpha[\nu} {\mathring \nabla X^{\lambda]}}_{\alpha}
\;,
\end{eqnarray*}
with $\mathring \nabla$   the Levi-Civita connection for the
Minkowski metric, and  $\det \bar g = \det (\bar g_{\rho
\sigma})$. The integration over $\partial_{\infty}\hyp$ is
taken, as usual, as a limit of integrals over spheres tending
to infinity. We let Greek indices run from 0 to 3, with $x^0=t$, and
$A^{[\mu\nu]}= \frac 1 2 (A^{\mu\nu}-A^{\nu\mu})$.  We note that under enough approximate parity (such as holds for the Kerr data we consider), the above integrals converge \cite{ChDelay}.

If $\hyp$, viewed as a hypersurface in  $(\mcM, \bar g)$, has
first and second fundamental forms $(g,K)$ satisfying
(\ref{af}) and (\ref{rt}), then for the appropriate choice of
$X$, the Hamiltonian yields the energy and momenta defined
earlier.  Indeed, if we let $X=\frac{\partial}{\partial
x^{\mu}}$, we get the energy-momentum one-form:  $H(\hyp,\bar
g, \frac{\partial}{\partial x^{\mu}})= p_{\mu}$, where $p_0=m$.
With $X=x^0 \frac{\partial}{\partial x^{i}}+ x^i
\frac{\partial}{\partial x^{0}}$, we get $H(\hyp, \bar g, X)=
mc^i$, and with $X=\epsilon_{ij\ell}x^j\frac{\partial}{\partial
x^{\ell}}$, we have $H(\hyp,\bar g, X)= J^i$.

We note that these Killing vectors correspond to \emph{Killing
Initial Data}, or \emph{KIDs}, for the Minkowski metric
\cite{KIDs}.  In Gaussian coordinates about $\hyp$, the
future-pointing unit normal is $\frac{\partial}{\partial t}$,
and for $Y$ tangent to $\hyp$, $N\frac{\partial}{\partial t}+Y$
is Killing for the Minkowski metric if and only if
$(D\Phi)^*(Y,N)=(0,0)$, where $D \Phi$ is the linearization of
the constraints operator at the Minkowski data.  We let
$\ourU^{\alpha\beta}(Y,N)$ correspond to
$X=N\frac{\partial}{\partial t}+Y$.

In fact the surface integrals for the Hamiltonian at finite
radii can be related to the constraints operator $\Phi(g,K)= (
\tau_i, \rho)$ by the following identity, where
$X=N\frac{\partial } {\partial t}+Y^i\frac{\partial}{\partial
x^i}$ corresponds to the KID $(Y,N)$, and $q$ is a quadratic
form in $(g_{ij}-\delta_{ij}, \frac{\partial g_{ij}}{\partial
x^k}, K_{ij})$ with coefficients uniformly bounded in terms of
bounds on $g_{ij}$ and $g^{ij}$:
\begin{eqnarray}
   \int_{\{x^0=0,r=R\}}
 \ourU^{\alpha\beta}(Y,N)dS_{\alpha\beta}&=& \int_{\{x^0=0,r=R_0\}}
 \ourU^{\alpha\beta}(Y,N)dS_{\alpha\beta} 
 \nonumber \\ && +\frac 1 {8\pi} \int_{\{x^0=0,R_0\le r\le R\}}
  \left(Y^i \tau_i + N \rho+q\right) d\mu_g\;.
  \nonumber \\ &&
\label{C4a}\end{eqnarray}

\section{Constructing $N$-Body Initial Data by Gluing}
\subsection{Many-Kerr Initial Data sets}
In this section we show how to attach a collection of $N$ Kerr
ends into one AE end while still solving the vacuum
constraints.  A special case of Theorem~\ref{main3} below, for
very special configurations (e.g., identical bodies placed
symmetrically about a center)  has earlier been established
in~\cite[Section~8.9]{ChDelay}.

For $0\leq a <b$, let $\Gamma(\mathbf  y, a,b)=\{ \mathbf  x \in
\mathbb{R}^3: a<|\mathbf  x - \mathbf  y|<b\}$ be a coordinate annulus
of inner radius $a$ and outer radius $b$ centred at $\mathbf  y$.

\begin{theorem}[Many-Kerr initial data sets]
 \label{main3}
For each $k=1, \ldots, N$, let $(E_k, g^k,K^k)$ be a Kerr
asymptotically Euclidean three dimensional end,  with
Poincar\'e charges
$$
 Q_k:=(m_k, \mathbf  p_k, m_k \mathbf  c_k, \mathbf 
J_k)
\;.
$$
Here the first entry is the ADM mass, the second is the ADM
momentum, $\mathbf  c_k$ is the centre of mass, and $\mathbf  J_k$  is
the total angular momentum. Suppose that the closures of the
annuli $\Gamma(\mathbf  c_k, 1,3)$ are \emph{pairwise disjoint},
and that  $\sum\limits_{k=1}^N (m_k, \mathbf  p_k)$ is time-like.

There exists $\epsilon_0>0$ such that for each
$0<\epsilon<\epsilon_0$ there is a vacuum asymptotically flat
initial data set $(E,g_{\epsilon},K_\epsilon)$ containing
isometrically
$$\bigcup\limits_{k=1}^N\left(\Gamma(\epsilon^{-1} \mathbf  c_k,
\epsilon^{-1},2\epsilon^{-1}), g^k,K^k\right)
 \;.
$$
\end{theorem}

\newcommand{\txe}{\tilde x_\epsilon}
\newcommand{\txek}{\tilde x_{\epsilon,k}}
\newcommand{\xek}{x_{\epsilon,k}}
\newcommand{\xe}{x_\epsilon}

\begin{proof}  Let $m_T=\sum\limits_{k=1}^N m_k$.  Since $m_T$ is non-zero, we can translate the origin of coordinates so that the center of mass $\sum\limits_{k=1}^N m_k \mathbf  c_k$ vanishes.  Let $B(r_0)\subset \mathbb{R}^3$ be a Euclidean ball of radius
$r_0=5+\max \{ |\mathbf  c_1|, \ldots, |\mathbf  c_N|\}$, centred at the origin.

We now construct a family of data on $\Omega_0:=B(r_0)\setminus \Big( \bigcup\limits_{k=1}^N \overline{\Gamma(\mathbf  c_k, 0, 1)} \Big)$.  Let $\chi$ be a smooth nondecreasing function so that $\chi(t)= 0$ for $t< 9/4$ and $\chi(t)=1$ for $t>11/4$.  Define $(\tilde g_{\epsilon},\tilde K_\epsilon)$ on $\Omega_0$ as follows:

\begin{enumerate}
\item[$\bullet$] On each $\Gamma(\mathbf  c_k,1,2)$, we
    let $(\tilde g_{\epsilon},\tilde K_\epsilon)$ be
    equal to $(g_{\epsilon,k}, K_{\epsilon,k})$,
    where $(g_{\epsilon,k}, K_{\epsilon,k})$ are initial data for a Kerr metric in
    Kerr-Schild coordinates with global charges
    $$
    Q_{\epsilon,k}=(\epsilon m_k, \epsilon \mathbf  p_k, \epsilon m_k \mathbf  c_k,\epsilon^2\mathbf  J_k)
    \;.
    $$

\item[$\bullet$] On each $\overline{\Gamma(\mathbf  c_k,2,3)}$,
    $$(\tilde g_{\epsilon}(\mathbf  x),\tilde K_\epsilon(\mathbf  x))=(1-\chi(|\mathbf  x-\mathbf  c_k|))
   (g_{\epsilon,k}, K_{\epsilon,k})  +\chi(|\mathbf  x-\mathbf  c_k|)
     (g_{\text{Eucl}},0)
     \;.
    $$
\newcommand{\xeext}{x_{\epsilon,\text{ext}}}

\item[$\bullet$]  On  $\Gamma(\mathbf  0,
    r_0-1, r_0)$, we let  $$(\tilde g_{\epsilon}(\mathbf  x),\tilde K_\epsilon(\mathbf  x))=(1-\chi(|\mathbf  x |-r_0+3))
     (g_{\text{Eucl}},0) +\chi(|\mathbf  x|-r_0+3)
    (g_{\epsilon, \rm{ext}}(\mathbf  x), K_{\epsilon, \rm{ext}}(\mathbf  x))
     \;.
    $$ where
    $ (g_{\epsilon, \rm{ext}}(\mathbf  x), K_{\epsilon, \rm{ext}}(\mathbf  x))$ is a Kerr metric with global charge
    $$
    Q_{\epsilon, \ext}=\Big(\underbrace{\epsilon (\sum_k m_k + \delta m)}_{=: m_{\epsilon, \ext}}, \epsilon(\sum_k \mathbf  p_k + \delta \mathbf  p),
   m_{\epsilon,\ext} \delta \mathbf  c, \epsilon^2
    \sum_k \mathbf  J_k + \epsilon \delta \mathbf  J\Big)
    \;,
    $$
    where
    $$
   \theta:= (\delta m, \delta \mathbf  p,  \delta \mathbf  c, \delta \mathbf  J)\in \Theta\subset \R^{10}
    \; .
    $$
    $\Theta$ is a compact convex set which will be specified below.
 \item[$\bullet$]  On $\overline{\Gamma(\mathbf  0, 0, r_0-1)}\setminus \bigcup\limits_{k=1}^N \overline{\Gamma(\mathbf  c_k, 1, 3)}$, let $(\tilde g_{\epsilon}, \tilde K_{\epsilon})=(g_{\rm{Eucl}}, 0)$.

\end{enumerate}

In what follows we will use the fact, which follows directly from the above definition and from the Kerr-Schild form of the Kerr metric, that for any $k \in
\N$,
we  have
    \bel{25V.2}
    \|(\tilde g_{\epsilon},\tilde K_\epsilon) -
    (g_{\text{Eucl}},0)\|_{C^k(\Omega_0)}\le C \epsilon
    \;,
    \ee
    for some constant $C=C(k)$.  We thus have $\Phi(\tilde g_{\epsilon},\tilde K_\epsilon)=O(\epsilon)$, and $\Phi(\tilde g_{\epsilon},\tilde K_\epsilon)=0$ in a neighborhood of $\partial \Omega$, where $\Omega=B(r_0)\setminus \bigcup\limits_{k=1}^N \overline{\Gamma(\mathbf  c_k, 0, 2)}$.

The goal is to modify $(\tilde g_{\epsilon},\tilde K_\epsilon)$ by adding a smooth deformation $(\delta g_{\epsilon},\delta  K_\epsilon)$ supported in $\overline \Omega$ so that $\Phi(\tilde g_{\epsilon}+\delta g_{\epsilon},\tilde K_\epsilon+\delta K_{\epsilon})=0$.  The proof proceeds in two stages, as we now describe.  Recall that the linearized constraints operator $D\Phi$ at the Minkowski data has a ten-dimensional cokernel $\mathcal K_0=\mbox{ker}(D\Phi^*)$ spanned by the KIDs mentioned earlier (Section \ref{ham}); this cokernel is an obstruction to solving the full system $\Phi(\tilde g_{\epsilon}+\delta g_{\epsilon},\tilde K_\epsilon+\delta K_{\epsilon})=0$ for $\delta g_{\epsilon}$ and $\delta K_{\epsilon}.$   However, $(\tilde g_{\epsilon},\tilde K_\epsilon)$ is close to the Minkowski data, and the constraints $\Phi(\tilde g_{\epsilon},\tilde K_\epsilon)$ vanish on a neighborhood of $\partial \Omega$, and so in particular $\Phi(\tilde g_{\epsilon},\tilde K_\epsilon)$ belongs to the appropriate weighted spaces used in \cite{ChDelay, CorvinoSchoen2}.   Therefore, by applying Theorem 2 of \cite{CorvinoSchoen2} or Corollary 5.11 of \cite{ChDelay}, we can at least solve the equation \emph{up to cokernel}: we can find, for each $\theta\in \Theta$ as above, a smooth deformation $(\delta g_{\epsilon}^{\theta}, \delta _{\epsilon}^{\theta})$ which is supported in $\overline \Omega$, satisfies $\| (\delta g_{\epsilon}^{\theta}, \delta K_{\epsilon}^{\theta})\|_{C^3(\overline \Omega)} \leq C\epsilon$ ($C$ is independent of $\theta$ and $\epsilon$), and also satisfies $\Phi(\tilde g_{\epsilon}+\delta g_{\epsilon}^{\theta},\tilde K_\epsilon+\delta K^{\theta}_{\epsilon})\in \zeta \mathcal K_0$, where $\zeta$ is a smooth weight function that vanishes on $\partial \Omega$.  For instance, in the notation of \cite{ChDelay}, $\zeta=\psi^2$, for $\psi$ a smooth function which near $\partial \Omega$ takes the form $\psi=e^{-s/d}$, where $d$ is a defining function for the boundary $\partial \Omega$, and where $s>0$.
The collection of $L^2(d\mu_e)$-projections of $\Phi(\tilde g_{\epsilon}+\delta g_{\epsilon}^{\theta},\tilde K_\epsilon+\delta K^{\theta}_{\epsilon})$ onto a basis for $\mathcal K_0$ defines  a map $ \Theta\ni\theta \mapsto \mathbb{R}^{10}$ which is continuous in $\theta$.  We show that this map vanishes for some $\theta$, and that will complete the proof; a rescaling then yields Theorem \ref{main3}.

We now focus on  computing these projections, which can be thought of as \emph{balance equations} as we shall see.  The primary tools for these computations are integration by parts together with the flux integrals (\ref{C4a}) and (\ref{adm}).
Note that the difference $|d\mu_e-d\mu_{\tilde g_{\epsilon}}|$ is $O(\epsilon)$, and that the Taylor expansion around the Minkowski data $(g_{\text{Eucl}}, 0)$ yields $$\Phi(\tilde g_{\epsilon}+\delta g_{\epsilon}^{\theta},\tilde K_\epsilon+\delta K^{\theta}_{\epsilon})=D\Phi(\tilde g_{\epsilon}+\delta g_{\epsilon}^{\theta}-g_{\rm{Eucl}},\tilde K_\epsilon+\delta K^{\theta}_{\epsilon})+ O(\epsilon^2).$$  As in Section \ref{ham}, integrating this quantity against the basis of KIDs we get the boundary integrands of (\ref{adm}), integrated over $\partial \Omega$, plus $O(\epsilon^2)$ terms.   Alternatively we can use an analogue of the integration-by-parts formula (\ref{C4a}) to derive the balance equations in a form analogous to (8.7) of~\cite{ChDelay}.   In either case, recall that on $\partial \Omega$ our solution up to cokernel agrees with Kerr data.  For such data, the differences between the limiting surface integrals (which give the charges) and  surface integrals at finite radius are given by the divergence theorem as in (\ref{C4a}). On the other hand, the Kerr data solves the vacuum constraints, so that the integrand over the annulus in (\ref{C4a}) reduces to $q$.  In the case of the data $(g_{\epsilon,\ext}, K_{\epsilon, \ext})$ under consideration, this results in a $O(\epsilon^2)$ term; note that we use Kerr-Schild coordinates, in which the metric components take the form $(\eta_{\mu\nu}+O(m_{\epsilon, \ext}))$.

Keeping all this in mind, we can now compute the balance equations.  Working first with  $e_{(1)}=(0,1)$, the KID
corresponding to the Minkowskian Killing vector $\frac{\partial}{\partial t}$ (cf. Section \ref{ham}), we obtain

\begin{eqnarray}
\frac{1}{16\pi}\langle e_{(1)}, \Phi(\tilde g_{\epsilon}+\delta g_{\epsilon}^{\theta},\tilde K_\epsilon+\delta K^{\theta}_{\epsilon})\rangle_{L^2(\Omega)} & = &
 \int_{\partial \Omega }{\mathbb U}^{\alpha \beta} (e_{(1)})dS_{\alpha\beta}+O(\epsilon^2)
 \nonumber
 \\
 & = &
 \epsilon\delta m +O(\epsilon^2)
 \;.
 \label{C4.1.1}
\end{eqnarray}
Next choosing $e_{(1+i)}=(\partial_i,0)$, for each $i=1,2,3$, to be the KID
corresponding to the Minkowskian Killing vector $\frac{\partial}{\partial x^i}$,
we calculate the balance equation to be
\begin{eqnarray}
\frac{1}{16\pi}\langle e_{(1+i)}, \Phi(\tilde g_{\epsilon}+\delta g_{\epsilon}^{\theta},\tilde K_\epsilon+\delta K^{\theta}_{\epsilon})\rangle_{L^2(\Omega)} & = &
 \epsilon\delta p^i +O(\epsilon^2)
 \;.
 \label{C4.1.2}
\end{eqnarray}
Now let $e_{(4+i)}$, for each $i=1,2,3$, be the KID corresponding to the
Minkowskian Killing vector $t\frac{\partial}{\partial x^i}+x^i\frac{\partial}{\partial t}$; at
$t=0$ this corresponds to $(Y,N)=(0,x^i)$. The boundary
integral around each $\mathbf  c_k$ is thus
\beaa
 \int_{\partial B({\mathbf  c_k},1)}{\mathbb U}^{\alpha \beta}(0,x^i) dS_{\alpha\beta}
 & = &
 \int_{\partial B({\mathbf  c_k},1)}{\mathbb U}^{\alpha \beta}(0,x^i-c^i_k+c^i_k) dS_{\alpha\beta}
\\
 & = &
 \ub{\int_{\partial B({\mathbf  c_k},1)}{\mathbb U}^{\alpha \beta}(0,x^i-c^i_k)
 dS_{\alpha\beta}}_{0+O(\epsilon^2)}
 +
 {\int_{\partial B({\mathbf  c_k},1)}{\mathbb U}^{\alpha \beta}(0,c^i_k)
 dS_{\alpha\beta}}
\\
 & = &
 c_k^i\ub{\int_{\partial B({\mathbf  c_k},1)}{\mathbb U}^{\alpha \beta}(0,1)
 dS_{\alpha\beta}}_{\epsilon m_k +O(\epsilon^2)}  +O(\epsilon^2)
\\
 & = &
  \epsilon m_k c^i_k + O(\epsilon^2)
 \;,
\eeaa
where the first integral in the second line vanishes, up to
quadratic terms in the metric, by definition of the centre of mass;
we have also used linearity of the global charges with respect
to the KIDs, and the fact that the ADM mass is the Poincar\'e
charge associated to the KID $(0,1)$. It follows that the
associated balance equation reads
\begin{equation}
\frac{1}{16\pi}\langle e_{(4+i)}, \Phi(\tilde g_{\epsilon}+\delta g_{\epsilon}^{\theta},\tilde K_\epsilon+\delta K^{\theta}_{\epsilon})\rangle_{L^2(\Omega)}  =
m_{\epsilon, \ext}\delta c^i -  \epsilon\underbrace{\sum_k m_k c^i_k}_0  +O(\epsilon^2)
 \;.
 \label{C4.1.3}
\end{equation}
Finally, let $\partial_l=\frac{\partial}{\partial x^\ell}$, and let $e_{(7+i)}$, $i=1,2,3$, be the KIDs corresponding
to the Minkowskian Killing vector $\epsilon_{ij\ell}x^j
\partial_\ell$; thus $(Y,N)=(\epsilon_{ij\ell}x^j
\partial_\ell,0)$. We then have
\beaa
 \int_{\partial B({\mathbf  c_k},1)}{\mathbb U}^{\alpha \beta}
 (\epsilon_{ij\ell}x^j\partial_\ell,0) dS_{\alpha\beta}
 & = &
 \int_{\partial B({\mathbf  c_k},1)}{\mathbb U}^{\alpha
 \beta}(\epsilon_{ij\ell}(x^j - c^j_k+c^j_k)
\partial_\ell,0) dS_{\alpha\beta}
\\
 & = &
 \ub{
 \int_{\partial B({\mathbf  c_k},1)}{\mathbb U}^{\alpha
 \beta}(\epsilon_{ij\ell}(x^j - c^j_k)
\partial_\ell,0)
 dS_{\alpha\beta}}_{\epsilon^2 J_k^i+O(\epsilon^2)}
\\
&& +
 {\int_{\partial B({\mathbf  c_k},1)}{\mathbb U}^{\alpha \beta}(\epsilon_{ij\ell}c^j_k
\partial_\ell,0)
 dS_{\alpha\beta}}
\\
 & = &
 \epsilon_{ij\ell}c^j_k\ub{\int_{\partial B({\mathbf  c_k},1)}{\mathbb U}^{\alpha \beta}(\partial_\ell,0)
 dS_{\alpha\beta}}_{\epsilon p^\ell_k +O(\epsilon^2)}  +O(\epsilon^2)
\\
 & = &
 \epsilon  \epsilon_{ij\ell} c^j_k p^\ell_k + O(\epsilon^2)
 \;,
\eeaa
We conclude that (recall that $Q_{\epsilon, \mbox{ext}}$ has angular momentum $\epsilon^2
    \sum_k \mathbf  J_k + \epsilon \delta \mathbf  J$)
\begin{equation}
\frac{1}{16\pi}\langle e_{(7+i)}, \Phi(\tilde g_{\epsilon}+\delta g_{\epsilon}^{\theta},\tilde K_\epsilon+\delta K^{\theta}_{\epsilon})\rangle_{L^2(\Omega)}   =
 \epsilon \delta J^i - \epsilon  \sum_k(\mathbf  c_k \times \mathbf  p_k)^i + O(\epsilon^2).
 \label{C4.1.4}
\end{equation}
We let $\Theta=(0,0,0,\sum\limits_k \mathbf  c_k \times \mathbf  p_k)+B_0$, where $B_0$ is a closed ball around the origin chosen so that $Q_{\epsilon, \ext}$ has time-like four-momentum.  We can invoke now the Brouwer fixed point theorem, in a way
similar to the proof of Theorem~8.1%
\footnote{Compare~\cite{ChDelayHilbert} for smoothing
arguments.} of~\cite{ChDelay}, to conclude that there exist
$\epsilon_0$ small enough so that the
right-hand-sides of the balance equations
\eq{C4.1.1}-\eq{C4.1.4} can be all made to vanish for all
$0<\epsilon<\epsilon_0$, which is the desired result.

\end{proof}

\subsection{Reduction to Kerr Asymptotics}
We recall the well-known result of
~\cite{ChDelay,CorvinoSchoen2} which states that any AE vacuum
end satisfying the Regge-Teitelboim condition (\ref{rt}) (with
time-like ADM four-momentum) can be deformed, outside of a
compact set, to a new vacuum initial data set such that  the
data agrees with that of a suitably chosen space-like slice of
a Kerr space-time, outside of a compact set.  We emphasize that
this deformation can be performed to preserve any given
pre-compact subset of the end; therefore  we can apply this
deformation to the data sets for each of our $N$ bodies,
preserving as large a compact subset of the original data as we
like.  We note for our main theorem that the four-momentum of
the resulting Kerr can be made as close to the original
four-momentum of the end as we like by performing the
deformation at larger coordinate radii.

Here we give a significant generalization of the above result,
which removes the Regge-Teitelboim assumption (\ref{rt}) in the
above gluing:

\begin{proposition} \label{kerr}
Let $(g,K)$ be AE vacuum initial data on the exterior $E$ of a
ball in $\mathbb R^3$ which has time-like ADM four-momentum
$(m, \mathbf  p)$.  Let $\epsilon>0$.  For sufficiently large $R$,
there is a vacuum initial data set $(\bar g, \bar K)$ on $E$ so
that on $E\cap \{ |\mathbf  x|\leq R\}$ we have  $(\bar g, \bar K)=(g,K)$, and so that on $\{
|\mathbf  x| \geq 2R\}$, $(\bar g, \bar K)$ is identical to data
from a space-like slice of a suitably chosen Kerr space-time.
If $( m+\delta m, \mathbf { p}+\delta \mathbf  p)$ is the four-momentum
of $(\bar g, \bar K)$, then $| \delta m|  < \epsilon$ and
$\left | \delta \mathbf  p\right| <\epsilon.$  If, moreover,
(\ref{rt}) holds, then also $|\delta \mathbf  c|<\epsilon$ and
$|\delta \mathbf  J|<\epsilon$.
\end{proposition}

\begin{proof} The proof is a minor modification of  that used in
\cite{ChDelay, CorvinoSchoen2} to prove the result for  the
case in which the condition  (\ref{rt}) holds.  In fact the
primary modifications which are needed have been introduced in
the proof of Theorem \ref{main3}.

Given a vacuum AE end, we let $\phi_R:A_1\rightarrow A_R$ be
the scaling $\phi_R(\mathbf  x)=R\mathbf  x$, and let $(g^R, K^R)=
(R^{-2} \phi_R^*g, R^{-1} \phi_R^*K)$.  We note that since
(\ref{rt}) may not hold, the decay might not be good enough to
make the centre of mass and angular momentum well-defined.
However, we have
\begin{eqnarray}
R \int\limits_{\{r=1\}}\sum\limits_{j,k}
(K^R_{jk}- (K^R)^{\ell}{}_\ell g^R_{jk})Y^j_i\nu^k d\sigma_e&=& R^{-1}\int\limits_{\{r=R\}}\sum\limits_{j,k}
(K_{jk}-K^{\ell}{}_{\ell}g_{jk})Y^j_i\nu^k d\sigma_e \nonumber \\
&=& O\Big(\frac{\log R}{R}\Big). \label{afjest}
\end{eqnarray}
Similarly,
\begin{equation}
R \int\limits_{\{r=1\}}
\Big[ \sum\limits_{i,j} x^\ell \left( g_{ij,i}-g_{ii,j}\right) \nu^j -
 \sum\limits_{i}\big(g_{ik}\delta^{k\ell} \nu^i- g_{ii}\nu^\ell \big) \Big] d\sigma_e =O\Big(\frac{\log R}{R}\Big). \label{afjest0}
\end{equation}
These two estimates follow from (\ref{bdyint}), together with the fact that $(g,K)$ solve the vacuum constraints.  Indeed the estimate $Q(h,K)=O(|\mathbf  x|^{-4})$ implies that\\ $(Y,N)\cdot Q(h,K)=O(|\mathbf  x|^{-3})$ for $(Y,N)=(Y_i, x^\ell)=O(|\mathbf  x|)$.  Now applying (\ref{bdyint}) and the constraints, we obtain $\mathcal{B}(R)=\mathcal{B}(R_0)+O(\log R)$.

We let $\epsilon= R^{-1}$, so that we can use some of the notation from the proof of Theorem \ref{main3}.  As above, we use a cutoff function to glue $(g^R,K^R)$ to Kerr data with charges $Q_{\epsilon, \ext}=\Big(\epsilon (m + \delta m), \epsilon(\mathbf  p + \delta \mathbf  p),
   \epsilon (m + \delta m) \delta \mathbf  c, \epsilon \delta \mathbf  J\Big)$ on the annulus $A_1$, so that in a neighborhood of the inner boundary of $A_1$ the data is identically $(g^R,K^R)$, and in a neighborhood of the outer boundary the data is identical to the Kerr data.  Again, we parametrize a family of such data with $\theta:= (\delta m, \delta \mathbf  p,  \delta \mathbf  c, \delta \mathbf  J)\in \Theta= B_0$, where $B_0$ is a closed ball around the origin, chosen so that the four-momentum of $Q_{\epsilon, \ext}$ is time-like.  Let $(\tilde g_{\epsilon}, \tilde K_{\epsilon})$ be the resulting glued data.  We note that $\|\tilde g_{\epsilon}-g_{\text{Eucl}}\|_{C^{\ell+1}(A_1)}+\|\tilde K_{\epsilon}\|_{C^{\ell}(A_1)}=O(\epsilon)$, that $\Phi(\tilde g_{\epsilon}, \tilde K_{\epsilon})=O(\epsilon)$, and that $\Phi(\tilde g_{\epsilon},\tilde K_{\epsilon})$ vanishes in a neighborhood of $\partial A_1$.   As above, we solve the vacuum constraints up to cokernel, so that $\Phi(\tilde g_{\epsilon}+\delta g_{\epsilon}^{\theta},\tilde K_\epsilon+\delta K^{\theta}_{\epsilon})\in \zeta \mathcal
     K_0$.

We again analyze the projection of $\Phi(\tilde
g_{\epsilon}+\delta g_{\epsilon}^{\theta},\tilde
K_\epsilon+\delta K^{\theta}_{\epsilon})$ onto $\mathcal K_0$, using the
notation from the proof in the last section, along with (\ref{afjest}) and (\ref{afjest0}):
\begin{eqnarray*}
\frac{1}{16\pi}\langle e_{(1)}, \Phi(\tilde g_{\epsilon}+\delta g_{\epsilon}^{\theta},\tilde K_\epsilon+\delta K^{\theta}_{\epsilon})\rangle_{L^2(A_1)} &=& \epsilon\delta m +O(\epsilon^2)\\
\frac{1}{16\pi}\langle e_{(1+i)}, \Phi(\tilde g_{\epsilon}+\delta g_{\epsilon}^{\theta},\tilde K_\epsilon+\delta K^{\theta}_{\epsilon})\rangle_{L^2(A_1)} & = &
 \epsilon\delta p^i +O(\epsilon^2)\\
\frac{1}{16\pi}\langle e_{(4+i)}, \Phi(\tilde g_{\epsilon}+\delta g_{\epsilon}^{\theta},\tilde K_\epsilon+\delta K^{\theta}_{\epsilon})\rangle_{L^2(A_1)}  &=& \epsilon\big[ \big(
m+\delta m)\delta c^i  +o(1)\big] \\
\frac{1}{16\pi}\langle e_{(7+i)}, \Phi(\tilde g_{\epsilon}+\delta g_{\epsilon}^{\theta},\tilde K_\epsilon+\delta K^{\theta}_{\epsilon})\rangle_{L^2(A_1)}   &=&
 \epsilon\big( \delta J^i + o(1)\big).
\end{eqnarray*}
The proof now follows from the Brouwer fixed point theorem as before, followed by scaling back to the annulus $A_R$.  \end{proof}

\subsection{Proof of Theorem \ref{main1}}
\begin{proof}We first apply Proposition \ref{kerr} to deform the data on each end $E_k$ to a new vacuum initial data set which agrees with the data on each $U_k$, and outside a compact set agrees with data from a suitably chosen space-like slice of a Kerr space-time.  A rescaling of all the metrics then
reduces the problem to one in which all of the initial data sets
are Kerrian outside of a Kerr-Schild coordinate-ball of radius
one. The result follows now by applying Theorem~\ref{main3}.
 \end{proof}

\section*{Acknowledgments}  
 The authors are grateful to Institut Mittag-Leffler (Djursholm, Sweden), for hospitality and financial support during the initiation of this paper.  PTC was supported in part
by the Polish Ministry of Science and Higher Education grant Nr
N N201 372736.  JC was partially supported by NSF grant DMS-0707317 and the Fulbright Foundation.  JI was partially supported by NSF grant PHY-0652903.

\end{document}